%% file: IROthogonal.tex
\documentclass[letterpaper, 10 pt, conference]{ieeeconf}

\IEEEoverridecommandlockouts                              
\usepackage{enumerate}
\usepackage{amsmath,amssymb,algorithm,cite,cases,bm,float,subfigure,graphicx,url}
\usepackage{pbox}
\newtheorem{defn}{Definition}[section]

\newtheorem{thm}{Theorem}[section]
\newtheorem{lem}{Lemma}[section]

\usepackage{amsmath,amssymb,algorithm,cite,subfigure,caption,cases,bm,float,subfigure,graphicx,url,color}
\usepackage{comment}
\usepackage{thm-restate, enumerate }

\definecolor{darkred}{RGB}{150,0,0}
\definecolor{darkgreen}{RGB}{0,150,0}
\definecolor{darkblue}{RGB}{0,0,200}

\input{commands}

\begin{document}
\title{\LARGE{\bf{
Isotropically Random Orthogonal Matrices:\\
Performance of LASSO and Minimum Conic Singular Values }}\vs{1pt}}
%
\author{
Christos Thrampoulidis and Babak Hassibi
%
\\
Department of Electrical Engineering, Caltech, Pasadena 
}

%

\maketitle
\begin{abstract} 

Recently, the precise performance of the Generalized LASSO algorithm for recovering structured  signals from compressed noisy measurements, obtained via i.i.d. Gaussian matrices, has been characterized.
The analysis is based on a framework introduced by Stojnic and heavily relies on the use of Gordon's Gaussian min-max theorem (GMT), a comparison principle on Gaussian processes. 
As a result, corresponding characterizations for other ensembles of measurement matrices have not been developed.
In this work, we analyze the corresponding performance of the ensemble of isotropically random orthogonal (i.r.o.) measurements. We consider the constrained version of the Generalized LASSO and derive a sharp characterization of its normalized squared error in the large-system limit. 
When compared to its Gaussian counterpart, our result analytically confirms the superiority in performance of the i.r.o. ensemble. Our second result, derives an asymptotic lower bound on the minimum conic singular values of i.r.o. matrices. This bound is larger than the corresponding bound on Gaussian matrices. To prove our results we express i.r.o. matrices in terms of Gaussians and show that, with some modifications, the GMT framework is still applicable.

\end{abstract}

\input{introduction}
\input{result}

\input{proof}
\bibliography{compbib}

\input{app}

\end{document}

%% file: commands.tex
\newcommand{\f}{\mathbf{f}}

\newcommand{\ellb}{\boldsymbol{\ell}}

\newcommand{\cone}{\operatorname{cone}}

\newcommand{\dthc}{\mathbf{d}_\h}

\newcommand{\las}{\frac{\la}{\sigma}}

\newcommand{\wt}{\tilde{\w}}

\newcommand{\rP}{\xrightarrow{P}}

\newcommand{\Thetab}{\boldsymbol\Theta}
\newcommand{\Phib}{\boldsymbol\Phi}
\newcommand{\qb}{\mathbf{q}}
\newcommand{\Qb}{\mathbf{Q}}

\newcommand{\R}{\mathbb{R}}

\newcommand{\ome}{\upsilon_{f,\x_0}(\beta,t;\g,\h,\vb)}

\newcommand{\smin}{\sigma_{\min}(\A;\Tc_f(\x_0))}

\newcommand{\lims}{\lim_{\sigma\rightarrow 0}}

\newcommand{\phiz}{\phi_0}

\newcommand{\Id}{\mathbf{I}}

\newcommand{\NSE}{\operatorname{NSE}}
\newcommand{\aNSE}{\operatorname{aNSE}}
\newcommand{\wNSE}{\operatorname{wNSE}}

\newcommand{\La}{\Lambda}

\newcommand{\fb}{\mathbf{f}}

\newcommand{\beq}{\begin{equation}}
\newcommand{\eeq}{\end{equation}}
\newcommand{\bea}{\begin{align}}
\newcommand{\eea}{{\end{align}}}

\newcommand{\vp}{\vspace{4pt}}

\newcommand{\temp}{\chi}

\newcommand{\mcsv}{\text{mCSV}}

\newcommand{\Delxf}{{\omega}^2_{f,\x_0}}
\newcommand{\DelxfR}{{\omega}_{f,\x_0}}

\newcommand{\G}{\mathbf{G}}

\newcommand{\X}{\mathbf{X}}
\newcommand{\A}{\mathbf{A}}

\newcommand{\w}{{\mathbf{w}}}

\newcommand{\ww}{{\mathbf{w}}}

\newcommand{\DS}{\mathcal{D}_f(\x_0)}

\newcommand{\x}{\mathbf{x}}
\newcommand{\ub}{\mathbf{u}}
\newcommand{\g}{\mathbf{g}}
\newcommand{\vb}{\mathbf{v}}
\newcommand{\bb}{\mathbf{b}}

\newcommand{\y}{\mathbf{y}}
\newcommand{\s}{\mathbf{s}}
\newcommand{\z}{\mathbf{z}}

\newcommand{\cb}{\mathbf{c}}  
\newcommand{\ab}{\mathbf{a}}

\newcommand{\h}{\mathbf{h}}

\newcommand{\wh}{\hat{\w}}

\newcommand{\Sc}{{\mathcal{S}}}

\newcommand{\Bc}{{\mathcal{B}}}

\newcommand{\Lc}{{\mathcal{L}}}

\newcommand{\Nn}{\mathcal{N}}
\newcommand{\Cc}{\mathcal{C}}

\newcommand{\Pro}{{\mathbb{P}}}
\newcommand{\E}{{\mathbb{E}}}

\newcommand{\paf}{\pa f(\x_0)}

\newcommand{\la}{{\lambda}}

\newcommand{\eps}{\epsilon}

\newcommand{\Tc}{\mathcal{T}}
\newcommand{\pa}{\partial}

\newcommand{\dt}{\text{{dist}}}

\newcommand{\vs}{\vspace}

\newcommand{\nn}{\nonumber}

%% file: introduction.tex
\section{Introduction}
\subsection{Setup}
Consider the classical problem of signal reconstruction of a \emph{structured} signal $\x_0\in\R^n$ from linear compressed and \emph{noisy} measurements $\y=\A\x_0 + \z\in\R^m$. 
Here $\A$ is the measurement matrix with compression rate $m/n<1$ and $\z$ is the noise vector. A standard method for recovering $\x_0$ is to solve a convex optimization program that enforces our prior knowledge about the distribution of the noise vector and the structure of the unknown signal. We model $\z$ as a zero-mean Gaussian vector with covariance matrix $\sigma^2\mathbf{I}$. Also assume $f:\R^n\rightarrow\R$ to be a \emph{convex} function that induces the structure of $\x_0$, e.g. $\ell_1$-norm for sparsity, nuclear norm for low-rankness, etc.. 
A popular algorithm in this direction is the Generalized $\ell_2^2$-LASSO that solves
\begin{align}\label{eq:intro_GL}
\hat\x^{\sigma} := \arg\min_{\x} \frac{1}{2}\|\y-\A\x\|_2^2 + \la f(\x),
\end{align}
for a regularization parameter $\la\geq 0$. We measure the performance of \eqref{eq:intro_GL} with the Normalized Squared Error (NSE):
\begin{align}\label{eq:intro_NSE}
\NSE(\sigma) := \| \hat\x^{\sigma} - \x_0 \|_2^2/\sigma^2,
\end{align}
and are interested in characterizing its behavior as a function of $n$, $m$, $f$, $\x_0$, $\sigma$ and $\lambda$.
To get a handle on this question, it is common to model the sampling matrix $\A$ as chosen at random from some ensemble. In particular,  two prominent models for the measurement matrix are:

\vspace{2pt}

\noindent(a)~\underline{Gaussian}: The entries of $\A$ are i.i.d. standard normal. This assumption is  primarily motivated by: (i) the  well-understood and remarkable properties of the gaussian ensemble, (ii) 
the so-called \emph{universality} property, i.e. many results turn out to hold true for matrices with i.i.d. entries drawn from a wide class of probability distributions.
\vspace{2pt}

\noindent(b)~\underline{Isotropically Random Orthogonal (i.r.o.)}: 
The matrix $\A$ is sampled uniformly at random from the manifold of row-orthogonal matrices satisfying $\A\A^T = \Id_m$.
Such orthogonal matrices are occasionally referred to as being ``Haar distributed". Matrices with orthogonal rows are often preferred in practice because their condition number is one and the do not amplify the noise. As a result they have superior noise performance, something we shall also observe in this paper too. Furthermore, certain classes of orthogonal matrices, such as Fourier, discrete-cosine and Hadamard allow for fast multiplication and reduced complexity.

%
%
%
\subsection{Background}
Understanding the reconstruction performance of \eqref{eq:intro_GL} has been a subject that has attracted enormous research attention over the past two decades or so. 
However, it is only recently that precise analysis in the noisy case has been developed.

\subsubsection{Noiseless Case}
In the noiseless case it has been shown \cite{Cha} that the unique solution to
$\min_{\{\x| \y=\A\x\}} f(\x)$  is the true vector $\x_0$ if  the number of measurements $m$ satisfy
\begin{align}\label{eq:intro_m>D}
m \gtrsim	\Delxf.
\end{align}
Here, $\Delxf$ is a geometric measure of the complexity of $f$ and $\x_0$, defined in Section \ref{sec:res}. \eqref{eq:intro_m>D} is \emph{precise} in the sense that the same number of measurements is also necessary \cite{TroppEdge}. This result is \emph{universal} over the measurement matrix $\A$ over both the Gaussian and the i.r.o. ensemble:  $\A$ appears in the optimality conditions 
only through its nullspace, which in both cases is an isotropically random subspace in $\R^n$ of dimension $n-m$.
\subsubsection{Noisy Case}
Most results in the noisy case are order-wise in the sense that they hold only up to unknown numerical constants. 

\vspace{3pt}

\noindent\underline{Gaussian Ensemble}:
Precise bounds on the NSE of the Generalized LASSO with Gaussian measurements have appeared only very recently. To the best of our knowledge, the first such results appear in \cite{DonohoSensitivity,MontanariLASSO} when $\ell_1$ regularization is used in \eqref{eq:intro_GL}. More recently, Stojnic introduced in \cite{StoLASSO} a novel framework, which is based on the use of Gordon's Gaussian min-max Theorem (GMT)  \cite[Lem.~3.1]{Gor}. The framework has proved to be powerful (also, see \cite{tight}) and has resulted  in simple, yet precise bounds on the $\NSE$ of the \emph{Generalized} LASSO \cite{StoLASSO,OTH,ICASSP,ISIT1}. Those results resemble \eqref{eq:intro_m>D} for the noiseless case. To get a flavor, consider the constrained version\footnote{From Lagrange duality there exists value of $\la$ in \eqref{eq:intro_GL} such that the two versions are equivalent.} of \eqref{eq:intro_GL} (C-LASSO) which solves 
\begin{align}\label{eq:intro_CL}
\hat\x := \arg\min_\x \|\y-\A\x\|_2 ~~\text{subject to}~~ f(\x)\leq f(\x_0).
\end{align}
It was shown in \cite{StoLASSO,OTH} that the NSE of \eqref{eq:intro_CL} under Gaussian measurements is upper bounded by
\begin{align}\label{eq:intro_flavor}
 \frac{\Delxf}{m-\Delxf}.
\end{align}
The bound is \emph{precise} since it is shown to be achieved with equality in the limit $\sigma\rightarrow 0$.

\vspace{3pt}

\noindent\underline{I.R.O. Ensemble }:
Unlike the noiseless case, in the noisy setting i.r.o. matrices exhibit different recovery performance than that of Gaussians. Using the replica method from statistical physics  and through extensive simulation results, \cite{replica1,replica2} derive expressions that characterize the NSE of \eqref{eq:intro_GL} and report that orthogonal constructions provide a \emph{superior} performance compared to their Gaussian counterparts. As mentioned in  \cite{replica1}, even though it provides a powerful tool for tackling hard analytical problems, the replica method still lacks mathematical rigor in some parts \cite{replica1}. As a follow up to these reports, and also driven by the fact that orthogonal constructions are easier to implement
%
%
 in practical applications \cite{replica2}, it is of interest to prove precise bounds on the achieved NSE; ones that would resemble those of \cite{StoLASSO,OTH,ISIT1} for Gaussian constructions. Towards this direction, Oymak and Hassibi showed in \cite{ACase} that the noisy performance of i.r.o. matrices is at \emph{at least as} good as that of Gaussians. To conclude this, they proved that the \emph{minimum conic singular value} ($\mcsv$) of the former can be no smaller than that of the latter. $\mcsv$s appears naturally as a measure of noise robustness performance (e.g.\cite[Cor.~3.3]{Cha}), thus,  the achieved NSE of i.r.o. can be no worse than that of Gaussians.
%
%
Adding to this, \cite{ACase} \emph{conjectures} a formula to bound the NSE of \eqref{eq:intro_CL} when $\A$ is i.r.o..

\subsection{Contribution}
We prove in Theorem \ref{thm:main} that when the measurement matrix $\A$ is i.r.o., then the NSE of \eqref{eq:intro_CL} in the high-SNR regime ($\sigma\rightarrow 0$) behaves precisely as\footnote{\eqref{eq:intro_flavor_res} holds for i.r.o. matrix $\A$ scaled such that $\A\A^T=n\Id_m$. This is to allow for a fair comparison with i.i.d. standard Gausian matrices for which $\E[\A\A^T]=n\Id_m$.}:
\begin{align}\label{eq:intro_flavor_res}
\frac{\Delxf}{m-\Delxf}\frac{n-\Delxf}{n}.
\end{align}
As is the case for the Gaussian ensemble (cf. \eqref{eq:intro_flavor}), we conjecture this to be the worst-case value of the NSE over all $\sigma$. 
Since $n-\Delxf<n$, when compared to \eqref{eq:intro_flavor}, our result implies the superiority  in performance of  the i.r.o. ensemble  when compared to the Gaussian one. 
 In particular, this  establishes rigorously the conjecture raised in \cite{ACase}. Our second result in Theorem \ref{thm:main2} derives a high-probability lower bound on the $\mcsv$ of i.r.o. matrices. The bound is seen to exceed the corresponding well-known bound for Gaussian matrices.

\subsection{Approach}
The set of techniques
available for dealing with i.r.o. matrices is limited compared to the variety of methods available for working with Gaussian matrices. 
Nonetheless, we are able to prove \eqref{eq:intro_flavor_res} based on a modification of the same framework \cite{tight} that led to corresponding results for the Gaussian case \cite{StoLASSO,OTH,ISIT1,ICASSP,Allerton}. As mentioned, the framework builds upon the GMT, a comparison lemma on Gaussian processes. In particular, \cite{StoLASSO,OTH} use the fact that $\|\ab\|_2=\max_{\|\ub\|\leq 1}\ub^T\ab$ to write \eqref{eq:intro_CL} as:
\begin{align}\label{eq:intro_trick}
\min_\x\max_{\|\ub\|_2\leq 1} \ub^T(\y-\A\x) ~~\text{subject to}~~ f(\x)\leq f(\x_0),
\end{align}
to which GMT is directly applicable.
 In contrast, when $\A$ is i.r.o., it is not at all obvious how to use GMT. To start with, there is no Gaussian matrix. The key idea here is to equivalently express an i.r.o. matrix  as:
\begin{align}
(\G\G^T)^{-1/2}\G,\nn
\end{align}
with $\G\in\R^{m\times n}$ having entries i.i.d. standard Gaussian and where $(\G\G^T)^{-1/2}$ is the inverse of the square-root of the positive definite (with probability one) $m\times m$ matrix $\G\G^T$.
Substituting in \eqref{eq:intro_CL}, the LASSO objective is closer but not yet quite of the form required by GMT.
In particular, the slick trick that led to \eqref{eq:intro_trick} is not enough here and additional ideas are required.  Using these we are able to bring \eqref{eq:intro_CL} into the desired format; the argument is sketched is Section \ref{sec:massage}.
Once this is done, what remains is to apply the framework of \cite{tight} to conclude with the desired.

%% file: result.tex
\section{Result}\label{sec:res}

\subsection{ Setup}

Let $\x_0\in R^n$, $\y= \A\x_0 + \sigma\vb\in \R^m$ and \emph{convex} $f:\R^n\rightarrow\R$. The constrained Generalized LASSO (C-LASSO) solves \eqref{eq:intro_CL}.
The reconstruction vector $\hat\x$ depends explicitly on $\A, f, \x_0$, and, implicitly on $\sigma, \vb$ through the measurement vector $\y$.
Define the NSE of \eqref{eq:intro_CL} as in \eqref{eq:intro_NSE}.

\subsubsection{Assumptions}\label{sec:ass}
The matrix $\A\in\R^{m\times n},m\leq n$ is modeled to have orthogonal rows $\A\A^T = \Id_m$, and the joint probability density of its elements remains unchanged when $\A$ is pre- and post- multiplied  by  any orthogonal matrices $\Phib\in\R^{m\times m}, \Thetab\in\R^{n\times n}$, i.e., $p(\Phib\A\Thetab) = p(\A)$. We say that $\A$ is \emph{i.r.o.}\footnote{
Different terminologies that appear in the literature to describe the same distribution include ``random $m$-frames in $\R^n$" and ``distributed according to the Haar measure on the Stiefel manifold", see \cite{tropp2012comparison}.
}. The noise vector $\vb$ has entries i.i.d. standard normal $\Nn(0,1)$, $f:\R^n\rightarrow\R$ is assumed convex and continuous, and,
$\x_0$ is not a minimizer of $f$. Popular regularizers include the $\ell_1$-norm, nuclear-norm, $\ell_{1,2}$-norm etc. (please refer to \cite{Cha,TroppEdge} for further examples).

\subsubsection{Large system limit}\label{sec:large}
Our results hold in an asymptotic regime in which the problem dimensions grow to infinity. We consider  a sequence of problem instances $\{\A,\vb,\x_0,f\}_{m,n}$ as in \eqref{eq:intro_CL} indexed by $m$ and $n$ such that both $m,n\rightarrow\infty$. In each problem instance, $\A, \vb$ and $f$ satisfy the assumptions of Section \ref{sec:ass}. Furthermore, $\hat\x$ and $\NSE(\sigma)$ denote the output of \eqref{eq:intro_CL} and the corresponding NSE. To keep notation simple, we avoid introducing explicitly the dependence of variables on the problem dimensions $m,n$. 

\subsubsection{NSE}
Define the \emph{worst-case} and \emph{asymptotic} NSE as
$
\wNSE := \sup_{\sigma>0} \NSE(\sigma),
$ 
and
$
\aNSE := \lim_{\sigma\rightarrow 0} \NSE(\sigma),
$ respectively.
The importance of studying the $\aNSE$ stems from the fact that $\wNSE=\aNSE$ in several cases (including C-LASSO for Gaussian measurements, also see \cite{ISIT1,OTH,Yu}). Theorem \ref{thm:main} precisely characterizes $\aNSE$. We conjecture that the same expression predicts $\wNSE$.

\subsubsection{Terminology}

\begin{defn}[Tangent cone] Consider $f:\R^n\rightarrow\R$, $\x_0\in\R^n$ and its set of descent directions  $\DS:=\{\w\in\R^n\big| f(\x_0+\w)\leq f(\x)\}$. The tangent cone of $f$ at $\x_0$ is defined as
$T_f(\x_0):=\operatorname{Cl}(\operatorname{cone}(D_f(\x))\nn
$.
, where $\operatorname{cone}(\cdot)$ and $\operatorname{Cl}(\cdot)$ return the conic hull and the closure of a set, respectively.
\end{defn}
\begin{defn}[Gaussian width] Let $\h\in\R^n$ have i.i.d. standard normal entries. The Gaussian width of the tangent cone of $f$ at $\x_0\in\R$ is defined as,
\beq
\DelxfR =\E_\h\Big[\sup_{\substack{\w\in\Tc_f(\x_0) , \|\w\|_2=1 }}~\h^T\w\Big]\nn.
\eeq
\end{defn}
The Gaussian width is a geometric measure of the size of the tangent cone. It is similarly defined for any set; the definition above is specific to our application. 
 Please refer to \cite{Cha,TroppEdge} for detailed discussions on its role in asymptotic convex geometry and on its properties.
 We also need the definition of the \emph{minimum conic singular value} ($\mcsv$) of a matrix $\A$. This can be defined for any cone in $\R^n$. To avoid introducing extra notation, we only define it with respect to the tangent cone of a function.
\begin{defn}[Minimum conic singular value] Let $\A\in\R^{m\times n}$. The minimum conic singular value of $\A$ with respect to the tangent cone of $f$ at $\x_0\in\R^n$ is defined as,
\beq\nn
\smin = \inf_{\substack{\x\in\Tc_f(\x_0) , \|\x\|_2=1}}\|\A\x\|_2.
\eeq
\end{defn}
Note that $\sigma_{\min}(\A;\R^n)$ is the minimum singular value of $\A$.

\subsection{Results}\label{sec:main}

Our results hold in the asymptotic \emph{linear regime}, where $m,n$ and $\Delxf$ all grow to infinity such that $m/n\rightarrow\delta\in(0,1)$ and $(1-\eps)m>\Delxf>\eps m$ for some constant $\eps>0$. 
In particular, assume the setup as in Section \ref{sec:large} under this linear regime. Also, let $\A\in\R^{m\times n}$ be distributed i.r.o.  
\begin{thm}[C-LASSO]\label{thm:main}
Consider \eqref{eq:intro_CL} and let
$$\aNSE:=\lim_{\sigma\rightarrow 0}{\|\hat\x-\x_0\|_2^2}/{\sigma^2}.$$
The following limit holds in probability
\begin{align}
\lim_{n\rightarrow\infty} \frac{\aNSE}{n}= \frac{\Delxf}{m-\Delxf}\frac{n-\Delxf}{n}. \nn
\end{align}
\end{thm}


\vp

\begin{thm}[Minimum Conic Singular Value]\label{thm:main2}
Denote
$$
\temp := \frac{\sqrt{n-\Delxf}}{\sqrt{n-m}}\frac{\sqrt{m}}{n}-\frac{\DelxfR}{n}
$$
and $\rho:=\DelxfR/\chi+n-m$.
For  all $\zeta>0$, with probability 1 in the limit $n\rightarrow\infty$, $\smin$ is lower bounded by
\begin{align}\nn
\sqrt\frac{m+\rho^2\temp^2-2\rho\temp\DelxfR-\rho\temp^2(n-m)}{m+\rho} - \zeta.
%
%
\end{align}
\end{thm}

\subsection{Remarks}\label{sec:rem}
\subsubsection{C-LASSO}
$\\$
~~~\emph{Comparison to Gaussian case:}
For an i.i.d Gaussian matrix with entries of variance $1/n$, it has been shown in \cite{OTH} that $\aNSE/n\approx \Delxf/(m-\Delxf)$. This is strictly greater than the expression of Theorem \ref{thm:main}, proving that the i.r.o. ensemble has strictly superior noise performance. Note that when $\Delxf<m\ll n$, the  two formulae are close to each other. This agrees with the fact that the entries of a very ``short" i.r.o. matrix are effectively independent for many practical purposes \cite{jiang2006many}. Finally, observe that both bounds approach infinity as the number of measurements $m$ approaches $\Delxf$. Of course, this agrees with the phase transition in the noiseless case (cf. \eqref{eq:intro_m>D}) which is same for both ensembles.

\vspace{2pt}

\emph{Interpretation}:
As seen the formula of Theorem \ref{thm:main} closely resembles the corresponding results for the Gaussian case. Thus, most of the remarks made for the Gaussian case (e.g. \cite{ISIT1}) regarding the role of the involved parameters, the geometric nature of the bound and its generality directly transfer to our case. It is useful to remark that$\Delxf$ admits precise high-dimensional approximations either in closed-form, or numerically tractable, for a number of useful instances of $f$ and $\x_0$, e.g.\cite{Cha,TroppEdge,OTH}. For a mere illustration, for $f=\|\cdot\|_1$ and $\x_0$ $k$-sparse signal, $\Delxf\lesssim 2k(\log(n/k) + 1)$.

\vspace{2pt}

\emph{$\wNSE$}: We conjecture that $\wNSE=\aNSE$. In this case, Theorem \ref{thm:main} would prove a tight upper bound on $\NSE(\sigma)$ for any $\sigma$. Simulation results in Figure \ref{fig:NSE} support the claim.

 \begin{figure}[!t]
\centering
  \includegraphics[width=1\linewidth]{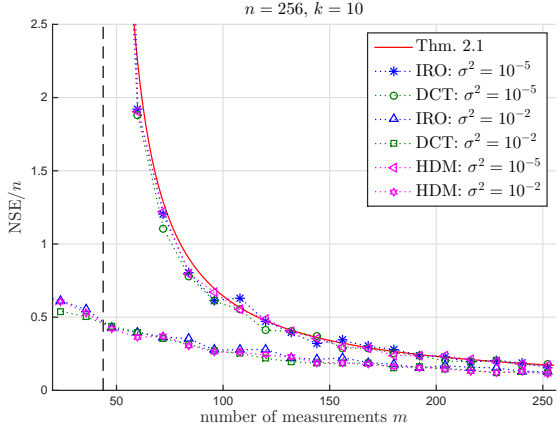}
    \label{fig:NSE}
  \caption{\footnotesize{Illustration of Theorem \ref{thm:main} for $f=\|\cdot\|_1$ and $\x_0\in\R^{256}$ a $10$-sparse vector. Simulation results support the claim that $\aNSE=\wNSE$. Furthermore, randomly sampled Discrete Cosine Transform (DCT) and Hadamard (HDM) matrices appear to have same $\NSE$ performance as i.r.o. matrices. Measured values of the $\NSE$ are averages over 25 realizations. }}
  \label{fig:NSE}
\end{figure}

\vspace{2pt}

\emph{Universality}: \cite{ACase} shows numerical evidence that partial Discrete Cosine Transform (DCT) matrices obtained by \emph{randomly} sampling $m$ rows of the DCT matrix without replacement, and similarly sampled Hadamard (HDM) matrices exhibit the same $\NSE$ performance as the i.r.o. ensemble. Our simulations in Figure \ref{fig:NSE} confirm this and, thus, Theorem \ref{thm:main} appears to predict the NSE of random DCT and HDM matrices as well. Understanding of the behavior of such ensembles is of great practical importance due to their favorable attributes\cite{replica2}.

 \begin{figure}[!t]
\centering
  \includegraphics[width=1\linewidth]{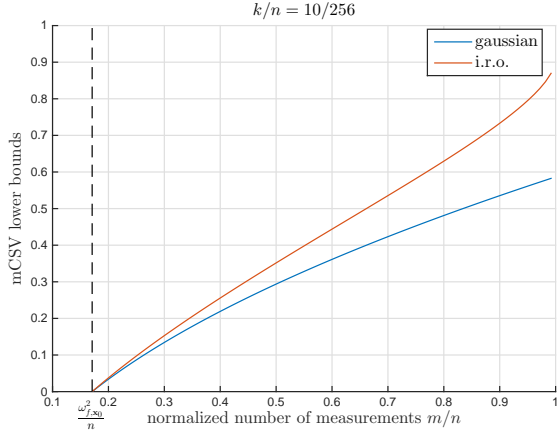}
    \label{fig:mcsv}
  \caption{\footnotesize{Illustration of Theorem \ref{thm:main2}. The bound exceeds the corresponding bound for Gaussian matrices. We have chosen $f=\|\cdot\|_1$ and $\x_0\in\R^{n}$ a $k$-sparse vector.  }}
  \label{fig:mcsv}
\end{figure}

\vp

\subsubsection{Minimum conic singular value}
$\\$
\emph{Comparison to Gaussian case:} A standard application of GMT shows that the $\mcsv$ of a matrix with i.i.d. entries $\Nn(0,1/n)$ is lower bounded by $(\sqrt{m}-\DelxfR)/\sqrt{n}$, e.g. \cite[Cor.~3.3]{Cha}. The bound of Theorem \ref{thm:main2} on the $\mcsv$ of an i.r.o. exceeds that, which is a strong indication that i.r.o. matrices are \emph{strictly} better conditioned than corresponding Gaussian ones. See Figure \ref{fig:mcsv} for an illustration. 

\emph{Sanity test:} When $\Delxf< m\ll n$, the entries of the i.r.o. behave almost as if they are independent\cite{jiang2006many}. As expected, then, in this regime the bound of Theorem \ref{thm:main2} approaches $(\sqrt{m}-\DelxfR)/\sqrt{n}$, which coincides with the bound on Gaussians. On the other hand, when $m=n$, it can be seen that, as expected, the expression of Theorem \ref{thm:main2} approaches one.


\emph{Tightness:} 
Theorem \ref{thm:main2} provides no guarantees on the exactness of the derived lower bound. This is also the case for the corresponding result on the $\mcsv$ of Gaussian matrices. Proving (or disproving) the exactness of the bounds is an open research problem.

\emph{General cones:} Of course, the bound of Theorem \ref{thm:main2} holds for the minimum singular value of $\A$ with respect to \emph{any} cone, not necessarily a tangent cone or even a convex cone. One just needs to replace $\DelxfR$ with the Gaussian width of the corresponding cone. Also, a non-asymptotic version of Theorem \ref{thm:main2} is possible, and will be included in the extended version of the paper.

%

%% file: proof.tex
\section{Proof Outline}\label{sec:proof}
%

In this section, we outline the main steps of the proof. We focus on Theorem \ref{thm:main}. The proof of Theorem \ref{thm:main2} follows along the same ideas and is only briefly discussed in the Appendix. Due to space considerations we limit our attention to showing the steps and modifications required to apply GMT in the case of i.r.o. matrices. In contrast to this part of the proof, which involves several new ideas, after we have transformed the problem into one where the GMT framework is applicable, then the rest is along the lines of \cite{StoLASSO,OTH,ICASSP,ISIT1}. This latter part and some technical details not discussed here are deferred to the Appendix and the extended version of the paper.
We re-write \eqref{eq:intro_CL} by changing the decision variable to be the error vector $\ww:=\x-\x_0$:
\begin{align}\label{eq:GL}
\hat\ww := \min_{\w\in\DS}\|\A\ww - \sigma\vb \|_2.
\end{align}
We evaluate the limiting behavior $\lim_{\sigma\rightarrow 0}\|\hat\ww\|^2/\sigma^2$. Throughout, we write $\|\cdot\|$ instead of $\|\cdot\|_2$.

 
\subsection{Formulation in terms of Gaussians}\label{sec:ir2G}
We begin with a simple lemma that provides a simple characterization of i.r. orthogonal matrices in terms of Gaussians. Let $\X^{1/2}$ denote a square-root of a matrix $\X\in\R^{m\times m}$, and $\X^{-1/2}$ its inverse (if it exists). Also, for random variables $x$ and $y$ with the same distribution, we write $x \sim y$.
\begin{lem}[I.r. orthogonal matrices]\label{lem:ir1}
Let $\G\in\R^{m\times n}$ have entries i.i.d. $\Nn(0,1)$. Then the matrix
$
\A = (\G\G^T)^{-1/2}\G
$
is a $m\times n$ i.r. orthogonal matrix.
\end{lem}
\begin{proof}
It can be readily confirmed that \mbox{$\A\A^T=\Id_m$}. We need to prove that the distribution of $\A$ remains invariant after pre- and post- multiplication with orthogonal matrices of appropriate sizes. 
Let \mbox{$\Phib\in\R^{n\times n}$}, \mbox{$\Thetab\in\R^{m\times m}$} be any orthogonal matrices. 
First,
\mbox{
$
\A\Thetab\sim(\G\G^T)^{-1/2}\G\Thetab = ((\G\Thetab)(\G\Thetab)^T)^{-1/2}\G\Thetab.
$
}
Recall that the Gaussian distribution is invariant under orthogonal transformations, i.e. $\G\sim\G\Thetab$, to conclude from the above that $\A\Thetab\sim\A$.
Next, $\G\sim\Phib\G$. Also, it can be directly verified that $\Phib(\G\G^T)^{-1/2}\Phib$ is the inverse of a square-root of of $\Phib\G\G^T\Phi$. With these,
$
\A\sim((\Phib\G)(\Phib\G)^T)^{-1/2}\Phib\G = \Phib(\G\G^T)^{-1/2}\G =  \Phib\A.
$
\end{proof}
Next, we use Lemma \ref{lem:ir1} to write the objective function in \eqref{eq:GL} in terms of Gaussian matrices.
\begin{lem}[LASSO Objective] \label{lem:ir2}
Assume $\A\in\R^{m\times n}$ is i.r. orthogonal and $\vb\in\R^m$ is standard Gaussian, independent of each other. Then, for any $\w\in\R^n$,
\mbox{$
(\A\w - \sigma\vb) \sim (\G\G^T)^{-1/2}\G(\sigma\qb-\w),
$}
where $\G\in\R^{m\times n}$ and $\qb\in\R^{n}$ have entries i.i.d. $\Nn(0,1)$ and are independent of each other.
\end{lem}
\begin{proof}
Let $\A,\G,\vb,\qb$ as in the statement of the Lemma.

For any row-orthogonal $\Qb\in\R^{m\times n}$, $\vb\sim\Qb\qb$. Furthermore, provided that $\qb$ is independent of the distribution of $\Qb$, the same is then true for $\vb$. Hence, letting $\Qb=\A$, we have $(\A\w - \sigma\vb)\sim\A(\w-\sigma\qb)$. Apply Lemma \ref{lem:ir1} to conclude with the desired.
\end{proof}


\subsection{Convex Gaussian min-max Theorem}\label{sec:GMT}
We get a handle on \eqref{eq:GL} and its optimal value via analyzing a different and simpler optimization problem, which we call \emph{Auxiliary Optimization} (AO) problem, as in \cite{tight}. The machinery that allows this relies on  Gordon's Gaussian min-max theorem (GMT) \cite[Lem.~3.1]{Gor}. In fact, we require a stronger version of the GMT that can be obtained when accompanied with additional \emph{convexity} assumptions that are not present in its original formulation. The fundamental idea is attributed to Stojnic \cite{StoLASSO}. \cite{tight} builds upon this and derives a concrete and somewhat extended statement of the result in \cite[Thm.~II.1]{tight}. Please refer to \cite{tight} for a discussion.
on the GMT, the role of convexity, and, the differences between \cite[Lem.~3.1]{Gor}, \cite{StoLASSO} and \cite[Thm.~II.1]{tight}. 
 We summarize the main idea of \cite[Thm.~II.1]{tight} in the next few lines.
Let $\G\in\R^{m\times n},\g\in\R^m,\h\in\R^n$ have entries i.i.d. Gaussian; $\Sc_\ab\subset\R^n$,$\Sc_\bb\subset\R^m$ be convex compact sets, and $\psi:\Sc_\ab\times\Sc_\bb\rightarrow\R$ be convex-concave and continuous. Consider the two min-max problems in \eqref{eq:Phi} and \eqref{eq:phi} which we refer to as \emph{Primary Optimization} (PO) and \emph{Auxiliary Optimization} (AO), respectively:
\begin{align}\label{eq:Phi}
\Phi(\G) := \min_{\ab\in\Sc_\ab} \max_{\bb\in\Sc_\bb} \bb^T\G\ab + \psi(\ab,\bb),
\end{align}
\begin{align}{\label{eq:phi}
{\phi(\g,\h) := \min_{\ab\in\Sc_\ab} \max_{\bb\in\Sc_\bb} \|\ab\|\g^T\bb - \|\bb\|\h^T\ab + \psi(\ab,\bb).}}
\end{align}
Then, for any $\mu\in\R, t>0$:
$$
\Pro\left(|\Phi(\G) - \mu|>t\right) \leq 2 \Pro\left(|\phi(\g,\h) - \mu|>t\right).
$$
In words, if the optimal cost of the (PO) concentrates to some value $\mu$, the same is true for the optimal cost of the (AO). 
Assuming a setup in which the problem dimensions $m,n$ grow to infinity it is shown in \cite{tight} that if $\phi(\g,\h)$ converges in probability to deterministic value $d_*$, then, so does $\Phi(\G)$. What is more, if $\|\ab_*(\g,\h)\|$ converges to say $\alpha_*$, then under appropriate strong convexity assumptions on the objective of \eqref{eq:phi}, $\|\ab_*(\G)\|$  converges to the same value.    Here, we denote $\ab_*(\g,\h)$, $\ab_*(\G)$ for the minimizers in \eqref{eq:phi} and \eqref{eq:Phi}.

\subsection{Deriving the Auxiliary Optimization Problem}\label{sec:massage}
Using Lemma \ref{lem:ir2}, we  work with the following (probabilistically) equivalent formulation of \eqref{eq:GL}:
\begin{align}
\wh := \min_{\w\in\DS}\| (\G\G^T)^{-1/2}\G(\w - \sigma\qb) \|_2 \label{eq:GL2}
\end{align}
This brings a step closer to the framework of GMT, but not yet quite to the point that we can identify the desired format of the (PO) as described in \eqref{eq:Phi}.
The goal of this section is to complete this step.
We start by using the fact that for any $\ab\in\R^m$:
$
\|\ab\|=\max_{\|\bb\|\leq 1}{\bb^T\ab}.
$
In particular, the objective function in \eqref{eq:GL2} can be expressed as follows:
\begin{align}
& \max_{\|\bb\|\leq 1}{\bb^T(\G\G^T)^{-1/2}(\w - \sigma\qb)}=\nn
\\& \max_{\|(\G\G^T)^{1/2}\bb\|\leq 1}{\bb^T\G(\w - \sigma\qb)}= \max_{\|\G^T\bb\|\leq 1}{\bb^T\G(\w - \sigma\qb)}\nn
\end{align}
It can be checked that the above is equivalent to:
\begin{align}
\max_{\bb}\min_{\ellb}{\bb^T\G(\w - \sigma\qb-\ellb)}+\|\ell\| \nn
\end{align}
Now, we flip the order of max-min \cite[Cor.~37.3.2]{Roc70}\footnote{(i) the objective function above is continuous, convex in $\ellb$, and concave in $\bb$, (ii) the constraint sets are convex. We only need to worry about \emph{boundedness} of the constraint sets. Such steps require proper attention in general and are handled rigorously in the Appendix. 
}:
\begin{align}\nn
\wh = &\min_{\substack{\w\in\DS, \ellb}}~\max_{\substack{\bb}}\bb^T\G(\w-\sigma\qb-\ellb) +\|\ellb\|,
\end{align}
or, re-defining $\ellb:=\w-\sigma\qb-\ellb$:
\begin{align}
\wh = &\min_{\substack{\w\in\DS, \ellb}}~\max_{\substack{\bb}} \bb^T\G\ellb +\|\w-\sigma\qb-\ellb\| \label{eq:GL_app_finnn}.
\end{align}
This brings \eqref{eq:GL} in the desired format of a (PO) problem\footnote{To be precise, this requires a trivial modification of \eqref{eq:GL_app} since $\w$ does not appear in the bilinear form as in \eqref{eq:Phi}. This can be handled easily and similar extension can also be found in \cite[Lem.~5]{Foygel}. In particular, in view of \eqref{eq:Phi} identify in \eqref{eq:GL_app_fin}: $\psi( [\ellb,\w] , \bb) :=  \|\w-\qb-\ellb\|$ which is continuous and convex in $[\ellb,\w]$, as desired. Also, the constraint sets are convex. Please refer to the Appendix for compactness issues.}
,and, allows us to derive the corresponding (AO) problem:
\begin{align}
\wt(\g,\h,\qb) = &\arg\min_{\substack{\w\in\DS, \ellb}}~\max_{\substack{\bb}} \|\ellb\|\g^T\bb - \|\bb\|\h^T\ellb\nn \\&\qquad\qquad\qquad\qquad +\|\w-\sigma\qb-\ellb\| \label{eq:GO0}.
\end{align}
The rest of the proof analyzes \eqref{eq:GO0} with the goal of determining the limiting behavior of $\|\wt\|$ and is included in the Appendix. We just remark here on the assumption of the theorem that $\sigma\rightarrow 0$; this also provides a hint on the precense of the gaussian width of the tangent cone in the final result. When $\sigma\rightarrow 0$, it suffices to analyze a ``first-order approximation" to problem \eqref{eq:GO0} in which the feasible set $\DS$ is substituted by its conic hull, i.e. $\Tc_f(\x_0)$. Since the tangent cone captures the local behavior in the neighborhood of $\x_0$, the relaxation will be tight in the limit as $\|\hat\w\|_2\rightarrow 0$.  The idea is that in the limit $\sigma\rightarrow 0$, $\|\hat\ww\|$ is sufficiently small and the approximation tight. 

%% file: app.tex
\appendix

Here we include a detailed proof of Theorem \ref{thm:main}. In the last section, we provide a short overview of the proof of Theorem \ref{thm:main2} which follows along the same key ideas.

\subsection{Preliminaries}\label{sec:pre}
We rewrite \eqref{eq:intro_GL} in a more convenient format for the purposes of the analysis. In particular, we perform the following operations in the order in which they appear: (i) substitute $\y=\A\x_0+\sigma\qb$,  (ii) change the decision variable to the quantity of
interest, i.e. the normalized error vector $\w:=(1/\sigma)({\x-\x_0})$, (iii) move the constraint on $\w$ to the objective function by introducing a Lagrange multiplier $\la$, and, (iv) rescale by a factor of $\sigma$. Then,
\begin{align}\label{eq:GL_app}
\wh := \min_{\w}\max_{\la\geq 0}\|\A\w - \qb \|_2 + \frac{\la}{\sigma}(f(\x_0+\sigma\w)- f(\x_0)).
\end{align}

We will derive a precise expression for the limiting (as $n\rightarrow\infty$) behavior of $\lim_{\sigma\rightarrow 0}\|\wh\|_2$. Note that after the normalization of $\x-\x_0$ with $\sigma$ , it is not guaranteed that the optimal minimizer in \eqref{eq:GL_app} is bounded (think of $\sigma\rightarrow 0$). However, we will prove that in the regime of Theorem \ref{thm:main} this is indeed the case. Many of the arguments that we use in the analysis require boundedness of the constraint sets. 
To tackle this, we assume that $\wh$ is bounded by some large constant $K>0$ (with probability one over $\A,\qb$), the value of which to be chosen at the end of the analysis. Recall that at that point we will have a precise characterization of the limiting behavior of $\|\wh\|_2$, say $\alpha_*$. If $\alpha_*$ turns out to be independent on the value of $K$ which we started with, then we will assume that this starting value was strictly larger than $\alpha_*$.  Thus, in what follows, we let $K,\La,M,...$ denote such (arbitrarily) large positive quantities. 
Also, throughout the proof we write $\|\cdot\|$ instead of $\|\cdot\|_2$.

\subsection{The Primary Optimization (PO) }\label{sec:pre}
Using Lemma \ref{lem:ir2}, onwards we work with the following (probabilistically) equivalent formulation of \eqref{eq:GL_app}:
\begin{align}
\wh := \min_{\|\w\|\leq K}&\max_{\substack{\la\geq 0}}\| (\G\G^T)^{-1/2}\G(\w - \qb) \|_2 \nn
\\ &~~~~~~ +\las(f(\x_0+\sigma\w)-f(\x_0))\label{eq:GL_app_2}.
\end{align}
The goal of this section is to bring this in a format for which GMT is applicable. 
We start by using the fact that for any $\ab\in\R^m$:
$$
\|\ab\|=\max_{\|\bb\|\leq 1}{\bb^T\ab}.
$$
In particular, the first term in \eqref{eq:GL_app_2} can be expressed as follows; to shorten notation denote $\cb:=\w-\qb$:
\begin{align}
\| (\G\G^T)^{-1/2}\G\cb \| &= \max_{\|\bb\|\leq 1}{\bb^T(\G\G^T)^{-1/2}\G\cb}\nn
\\=& \max_{\|(\G\G^T)^{1/2}\bb\|\leq 1}{\bb^T\G\cb}\nn
\\
&= \max_{\|\G^T\bb\|\leq 1}{\bb^T\G\cb}\label{eq:lin0}
\\
&=\max_{\|\bb\|\leq \La}{\bb^T\G\cb} - \delta(\G^T\bb| \Bc^{n-1}),\label{eq:lin1}
\end{align}
In the last line above, $\delta(\ab | \Bc^{n-1})$ denotes the indicator function of the unit ball, i.e. takes the value $0$ if $\|\ab\|\leq 1$ and $+\infty$, otherwise. Also, we are allowed to assume that $\bb$ is bounded by some large $0\leq\La\leq\infty$, since the set of optima in \eqref{eq:lin0}   is a compact set ($\G^T$ has full column rank with probability one). 
It can be readily checked (or, see \cite{Roc70} that $\delta(\ab | \Bc^{n-1}) = \sup_{\ellb}\ab^T\ellb - \|\ellb\|$, for any $\ab\in\R^n$. Thus, continuing from \eqref{eq:lin1}:
\begin{align}
\| (\G\G^T)^{-1/2}\G\cb \| 
=\max_{\|\bb\|\leq\La}\inf_{\ellb}{\bb^T\G(\cb-\ellb)}+\|\ell\| \nn
\end{align}
As a final step, we will flip the order of max-min above. This is allowed by \cite[Cor.~37.3.2]{Roc70} since:  (i) the objective function above is continuous, convex in $\ellb$, and concave in $\bb$, (ii) the constraint sets are convex, (iii) the set constraining the maximization is bounded. Thus,
\begin{align}
\| (\G\G^T)^{-1/2}\G\cb \| 
=\inf_{\ellb}\max_{\|\bb\|\leq\La}{\bb^T\G(\cb-\ellb)}+\|\ell\| \nn.
\end{align}
We argue that the infimum above is achieved over a bounded set. Indeed, performing the maximization over $\bb$ above
$$\inf_{\ellb}\max_{\|\bb\|\leq\La}{\bb^T\G(\cb-\ellb)}+\|\ell\| = \inf_{\ellb}\La\|{\G(\cb-\ellb)}\|+\|\ell\|$$
The sub-level sets of the (continuous) objective function in the minimization on the right-hand side of the equation above are clearly bounded. Hence, by Weierstrass' Theorem \cite[Prop.~2.1.1]{Bertsekas} the set of minimum is nonempty and compact. We may thus assume there exists large but finite $N$ such that constraining the minimization over $\|\ellb\|\leq N$ does not increase the optimum. 
We may now substitute the above in \eqref{eq:GL_app_2} to conclude with:
\begin{align}
\wh = &\min_{\substack{\|\w\|\leq K \\ \|\ellb\|\leq N}}~\max_{\substack{\la\geq 0\\ \|\bb\|\leq \La }}\bb^T\G(\w-\qb-\ellb) +\|\ellb\| \nn
\\ &\quad\quad\quad\qquad\qquad+\las(f(\x_0+\sigma\w)-f(\x_0))\nn.
\end{align}
or, re-defining $\ellb:=\w-\qb-\ellb$ and appropriately adjusting $N$:
\begin{align}
\wh = &\min_{\substack{\|\w\|\leq K \\ \|\ellb\|\leq N}}~\max_{\substack{\la\geq 0\\ \|\bb\|\leq \La }}\bb^T\G\ellb +\|\w-\qb-\ellb\| \nn
\\ &\quad\quad\quad\qquad\qquad+\las(f(\x_0+\sigma\w)-f(\x_0))\label{eq:GL_app_fin}.
\end{align}
This brings \eqref{eq:GL_app} in the desired format for the application of GMT.
 In particular, identify $\psi( [\ellb,\w] , \bb) :=  \|\w-\qb-\ellb\|  + \max_{\la\geq 0}\las(f(\x_0+\sigma\w)-f(\x_0)) $ which is continuous and convex in $[\ellb,\w]$, as desired. This format is of course the same as in \eqref{eq:GL_app_finnn}, modulo the boundedness constraints which were not regarded in the main body of the paper.


\subsection{The Auxiliary Optimization (AO) for arbitrary $\sigma$}\label{sec:gor}
Let us write the (AO) problem as it corresponds to \eqref{eq:GL_app_fin}:
\begin{align}
&\wt(\g,\h,\qb) = \min_{\substack{\|\w\|\leq K \\ \|\ellb\|\leq N}}~\max_{\substack{\la\geq 0\\ \|\bb\|\leq \La }} \|\ellb\|\g^T\bb - \|\bb\|\h^T\ellb  \nn
\\ &~~~~+\|\w-\qb-\ellb\|+\las(f(\x_0+\sigma\w)-f(\x_0))\label{eq:GO1}.
\end{align}
Our goal in the rest of the section is to simplify \eqref{eq:GO1}. By massaging the objective functions and performing minimizations/maximizations when possible we eventually reach to an equivalent formulation, in which most optimizations are in terms of scalar variables instead of vectors. Two remarks are in place:

\noindent(a) We will need to flip the order of min-max several times; except if stated differently we apply \cite[Cor.~37.3.2]{Roc70}: here, constraint sets will always be convex and the objective function continuous. We only need to worry about convexity of the objective and boundedness of (at least one of) the constraint sets. 

\noindent(b) To keep notation short, we will often drop the set constraints over the optimization variables when clear from context. Recall that most of the constraints are just boundedness constraints by constants that can be chosen large.

\subsubsection{Maximizing over the direction of $\bb$}
This is easy to perform, note that $\max_{\|\bb\|=\beta}\g^T\bb=\beta\|\g\|_2, \beta\geq 0$.

\subsubsection{Minimizing over $\ellb$} First, let us argue briefly that we can ``push" the minimization over $\ellb$ on the right of the maximization: (i) it can be seen that after optimizing over the direction of $\bb$, the objective function in \eqref{eq:GO1} is convex in $\ellb$ (i) it is also concave in $\la,\beta$, and (iii) $\ellb$ is constrained in a bounded set. 

To be able to optimize over $\ellb$, we use the following trick. We will express the terms $\|\ellb\|$ and $\|\w-\qb-\ellb\|$ using the fact that:
\begin{align}\label{eq:sqrtx}
\sqrt{x} = \min_{p\geq 0} \frac{x}{2p} + \frac{p}{2}, \quad\forall x>0.
\end{align}
Also, note that the set of minima above is clearly bounded for bounded $x$. With these,
\begin{align*}
&\min_{\ellb} \beta(\|\ellb\|\|\g\|-\h^T\ellb) + \|\w-\qb-\ellb\| = \\ &\min_{\substack{0\leq p \leq P \\ 0\leq t \leq T}} \frac{p+t}{2} +
\frac{1}{2p}\|\qb-\w\|^2+ \\ 
&~~~~\min_{\ellb} \frac{1}{p}\left[\frac{t+\beta^2 p\|\g\|^2}{2t}\|\ellb\|^2 + (-\beta p \h + \qb - \w)^T\ellb \right],
%
\end{align*}
and the minimization over $\ellb$ contributes the term:
$$
-\frac{1}{2p}\frac{t}{t+\beta^2 p\|\g\|^2}\|-\beta p \h + \qb - \w\|^2.
$$

\subsubsection{Linearize $f$}
$f$ is continuous and convex, thus, we can express it in terms of its convex conjugate $f^*(\ub) = \sup_{\x}\ub^T\x - f(\x)$. In particular, applying \cite[Thm.12.2]{Roc70} we have $f(\x_0+\sigma\w)=\sup_{\ub}\x_0^T\ub + \sigma\ub^T\w-f^*(\ub)$. The supremum here is achieved at $\ub_*\in\partial f(\x_0+\sigma\w)$ \cite[Thm.~23.5]{Roc70}. Also, from \cite[Prop.~4.2.3]{Bertsekas}, $\cup_{\|\w\|\leq K}\partial f(\x_0+\sigma\w)$ is bounded. Thus, the set of maximizers $\ub_*$ is bounded and for some $0< M:=M(K) <\infty$, $\wt$ is given as the solution to
\begin{align}
&\max_{\substack{\la\geq 0\\ 0\leq\beta\leq \La \\ \|\ub\|\leq M }}\min_{\w,p,t} \frac{p+t}{2} + 
\frac{1}{2p}\|\qb-\w\|^2 + \la\ub^T\w  \nn
\\ &~~~~~-\frac{1}{2p}\frac{t}{t+\beta^2 p\|\g\|^2}\|-\beta p \h + \qb - \w\|^2
+ \las F(\ub)\label{eq:GO3},
\end{align}
where we have flipped the orders of min-max for $\w$ and $\ub$, and have denoted 
$$
F(\ub):= \ub^T\x_0-f^*(\ub)-f(\x_0).
$$

\subsubsection{Redefine variables} It will be convenient for the calculations to follow to redefine the variables $\beta$ and $t$ as follows:
$$
\beta := \beta p ,\quad t := t p \quad\text{ and }\quad \la:=\la p.
$$
It can be checked that with these changes, the optimization remains convex. 

\subsubsection{Minimizing over the direction of $\w$} 
Evaluating the squares in \eqref{eq:GO3} and after some algebra, it can be shown that the terms in which $\w$ appears are as follows:
\begin{align}\label{eq:dirw}
\frac{\beta^2\|\g\|^2}{2(\beta^2\|\g\|^2+t)}\|\w\|^2 - ({\tilde\fb} - \la\ub)^T\w, 
\end{align}
where 
\begin{align}\label{eq:f}
{\tilde\fb} := \left( -\frac{\beta t}{\beta^2\|\g\|^2+t} \h + \frac{\beta^2\|\g\|^2}{\beta^2\|\g\|^2+t}\qb \right),
\end{align}
which has entries i.i.d. Gaussians of zero mean and standard deviation 
\begin{align}\label{eq:sf}
\sigma_{\tilde\fb}:=\sigma_{\tilde\fb}(\beta,t) := \frac{\beta\sqrt{t^2+\beta^2\|\g\|^4}}{\beta^2\|\g\|^2+t}.
\end{align}
Fix the norm of $\|\w\|=\alpha$. Optimizing over the direction of $\w$ the second term in \eqref{eq:dirw} gives  $-\alpha\|\tilde\f-\la\ub\|$.

\subsubsection{Minimize over $p$}
Overall, the min-max problem in \eqref{eq:GO1} has reduced itself to:
\begin{align}
&\max_{\substack{\la\geq 0\\ 0\leq\beta\leq \La \\ \|\ub\|\leq M }}\min_{\alpha,p,t}\Big\{
\frac{1}{2p}\Big( t +  \|\qb\|^2 - \frac{t}{\beta^2\|\g\|^2+t}\|\beta\h-\qb\|^2 + \nn \\ 
&\quad\frac{\beta^2\|\g\|^2}{\beta^2\|\g\|^2+t}\alpha^2 -2\alpha \|{\tilde\fb} - \la\ub\|  + 2\las F(\ub) \Big)+ \frac{p}{2}  \Big\}\nn =\\
&\max_{\substack{\la\geq 0\\ 0\leq\beta\leq \La \\ \|\ub\|\leq M }}\min_{\alpha,t}
\Big( t +  \|\qb\|^2 - \frac{t}{\beta^2\|\g\|^2+t}\|\beta\h-\qb\|^2 + \nn \\ 
&\quad\frac{\beta^2\|\g\|^2}{\beta^2\|\g\|^2+t}\alpha^2 -2\alpha \|{\tilde\fb} - \la\ub\|  + 2\las F(\ub) \Big)^{1/2}.
\label{eq:GO4}
\end{align}
In yielding the equality above, we have applied \eqref{eq:sqrtx}.

\subsubsection{Redifine $\la$}
It is convenient to redefine $\la$ as 
$
\la:=\la/\sigma_{\tilde\fb}.
$
Let $\fb$ denote standard  i.i.d. Gaussian vector, such that $\tilde\fb\sim\sigma_{\tilde\fb}\fb$. With these, we can express $\wt$ as the solution to:
\begin{align}&
\max_{\substack{\la\geq 0\\ 0\leq\beta\leq \La \\ \|\ub\|\leq M }}\min_{\alpha,t}
\Big( t +  \|\qb\|^2 - \frac{t}{\beta^2\|\g\|^2+t}\|\beta\h-\qb\|^2 + \nn \\ 
&\quad\frac{\beta^2\|\g\|^2}{\beta^2\|\g\|^2+t}\alpha^2 -2\sigma_{\tilde\fb} (\alpha \|{\fb} - \la\ub\|  - 2\las F(\ub) ) \Big).
\label{eq:GO5}
\end{align}
Note that we have essentially considered the square of \eqref{eq:GO5}. Let us denote the optimal cost of \eqref{eq:GO5} above as $\phi(\sigma):=\phi(\sigma;\g,\h,\qb,\fb)$.

 \subsection{The Auxiliary Optimization in the limit $\sigma\rightarrow 0$}

\cite[Thm.~II.2]{tight} relates $\|\wt\|$ to $\|\wh\|$, under appropriate assumptions. Also, recall that  we wish to characterize $\lim_{\sigma\rightarrow 0}\|\wh\|$. Thus, in view of \eqref{eq:GO5} we wish to analyze the problem 
$$
\phiz:=\phiz(\g,\h,\qb,\fb):=\lim_{\sigma\rightarrow 0}\phi(\sigma;\g,\h,\qb,\fb).
$$
 In \eqref{eq:GO5}, from Fenchel's inequality:
\begin{align}\label{eq:conj}
F(\ub) = \ub^T\x_0 - f^*(\ub)-f(\x_0)\leq 0.
\end{align}
With this observation, we prove in the next lemma that $\phi(\sigma;\g,\h,\qb,\fb)$ is non-decreasing in $\sigma$.
\begin{lem}\label{lem:decreasing}
Fix $\g,\h,\qb,\fb$ and consider $\phi(\cdot;\g,\h,\qb,\fb):(0,\infty)\rightarrow\R$ as defined in \eqref{eq:GO5}. $\phi(\sigma;\g,\h,\qb,\fb)$ is non-decreasing in $\sigma$.
\end{lem}
\begin{proof}
Denote $\Lc(\sigma,\alpha,t,\beta,\ub,\la)$ the objective function in \eqref{eq:GO5} and consider $0<\sigma_1<\sigma_2<\infty$. Let $\alpha^{(2)},t^{(2)}$ be an optimal solution to the min-max problem in \eqref{eq:GO5} for $\sigma_2$. Then, let $(\beta^{(1)},\ub^{(1)},\la^{(1)}) = \arg\max_{\beta,\ub,\la}\Lc(\sigma_1,\alpha^{(2)},t^{(2)},\beta,\ub,\la)$. Clearly,
$$
\phi(\sigma_1)\leq \Lc(\sigma_1,\alpha^{(2)},t^{(2)},\beta^{(1)},\ub^{(1)},\la^{(1)}).
$$
Using $1/\sigma_1>1/\sigma_2$ and  \eqref{eq:conj},
\begin{align*}
& \Lc(\sigma_1,\alpha^{(2)},t^{(2)},\beta^{(1)},\ub^{(1)},\la^{(1)}) \leq  \\ &\qquad\qquad\qquad\quad\Lc(\sigma_2,\alpha^{(2)},t^{(2)},\beta^{(1)},\ub^{(1)},\la^{(1)})
\end{align*}
But,
$$
\Lc(\sigma_2,\alpha^{(2)},t^{(2)},\beta^{(1)},\ub^{(1)},\la^{(1)}) \leq \phi(\sigma_2).
$$
Combine the above chain of inequalities to  conclude.
\end{proof}

%
In particular, when viewed as a function of $\kappa:=1/\sigma$, $\phi(\cdot;\g,\h,\qb,\f)$ is non-increasing. Thus,
 \begin{align}\label{eq:liminf}
 \phiz = \lims\phi(\sigma)= \lim_{\kappa\rightarrow\infty}\phi(\kappa)=\inf_{\kappa\geq 0}\phi(\kappa),
 \end{align}
Next, we argue that we can flip the order of min-max.
 The objective function in \eqref{eq:GO5} is continuous,  convex in $\kappa$, and, concave in $\la,\beta,\ub$. The constraint set on $\la$ appears to be unbounded, but, it can be checked from \eqref{eq:GO5} that the optimal value is in fact bounded. With this and \eqref{eq:liminf}, we get
%
\begin{align}&
\max_{\substack{\la\geq 0\\ 0\leq\beta\leq \La \\ \|\ub\|\leq M }}\min_{\alpha,t}\inf_{\kappa\geq 0}
\Big( t +  \|\qb\|^2 - \frac{t}{\beta^2\|\g\|^2+t}\|\beta\h-\qb\|^2 + \nn \\ 
&\quad\frac{\beta^2\|\g\|^2}{\beta^2\|\g\|^2+t}\alpha^2 -2\sigma_{\fb} \alpha \|{\tilde\fb} - \la\ub\|  + \kappa 2\sigma_{\tilde\fb} \la F(\ub)  \Big).
\nn
\end{align}
Recall \eqref{eq:conj} and the fact that equality is achieved iff $\ub\in\paf$ (e.g. \cite[Thm.~23.5]{Roc70}). Then, $\phi_0$ is given by
\begin{align}
&\max_{\substack{\la\geq 0\\ 0\leq\beta\leq \La \\ \ub\in\paf }}\min_{\alpha,t}
\Big( t +  \|\qb\|^2 - \frac{t}{\beta^2\|\g\|^2+t}\|\beta\h-\qb\|^2 + \nn \\ 
&\qquad\qquad\qquad\quad\frac{\beta^2\|\g\|^2}{\beta^2\|\g\|^2+t}\alpha^2 -2\sigma_{\fb} \alpha \|{\tilde\fb} - \la\ub\|   \Big). \nn
\end{align}
where we have assumed $\infty>M>\max_{\s\in\paf}\|\s\|$.
We can now optimize over $\la\ub$ (after appropriately flipping the order of min-max): $\min_{\la\geq 0, \ub\in\paf}\|\fb-\la\ub\| =  \dt(\f,\cone(\paf))$. Thus, we conclude with the (AO) for $\sigma\rightarrow 0$ taking the form:
\begin{align}
&\phi_0(\g,\h,\qb,\f) = \min_{\substack{0\leq\alpha\leq K}} \Lc(\alpha;\g,\h,\qb,\f), \label{eq:goropt}\\
&\Lc(\alpha;\g,\h,\qb,\f) := \min_{ 0\leq t\leq T} \max_{\substack{0\leq\beta\leq \La}}\Big\{ t +  \|\qb\|^2 -  \nn \\ 
&~\frac{t}{\beta^2\|\g\|^2+t}\|\beta\h-\qb\|^2 + \frac{\beta^2\|\g\|^2}{\beta^2\|\g\|^2+t}\alpha^2 -2\sigma_{\tilde\fb} \alpha \dthc \Big\}, \nn
\end{align}
where we have denoted $\dthc:=\dt(\f,\cone(\paf)) $.


It is now easy to optimize $\eqref{eq:goropt}$ over $\alpha$. We summarize the result in the following lemma.

\begin{lem}\label{lem:det}
In \eqref{eq:goropt}, fix $\g,\h,\qb$ and let $\wt:=\wt(\g,\h,\qb)$ be optimal. Denote,
$$
{\tilde\fb}:={\tilde\fb}(\beta,t) :=  -\frac{\beta t}{\beta^2\|\g\|^2+t} \h + \frac{\beta^2\|\g\|^2}{\beta^2\|\g\|^2+t}\qb ,
$$
\mbox{and $\upsilon(\beta,t):=\dt(\tilde\fb,\cone(\la\paf))$}.
 Then, 
\begin{align}\label{eq:wopt}
\|\wt\| = \frac{{\beta^2\|\g\|^2+t}}{\beta\|\g\|^2}\upsilon(\beta,t),
\end{align}
where $\beta,t$ are optimal solutions to the following optimization:
\begin{align}
&\max_{\substack{\La\geq\beta\geq 0}}\min_{T\geq t\geq 0}
\Big( t +  \|\qb\|^2 - \frac{t}{\beta^2\|\g\|^2+t}\|\beta\h-\qb\|^2 + \nn \\ 
&\quad\qquad-\frac{t^2+\beta^2\|\g\|^4}{\|\g\|^2(\beta^2\|\g\|^2+t)}\upsilon(\beta,t)\Big).
\label{eq:goropt}
\end{align}
\end{lem}
Note that $\tilde\fb\sim\sigma_{\tilde{\fb}}\fb$ where $\fb$ is standard i.i.d. Gaussian and
\begin{align*}
\sigma_{\tilde\fb}:=\sigma_{\tilde\fb}(\beta,t) := {\beta\sqrt{t^2+\beta^2\|\g\|^4}}/({\beta^2\|\g\|^2+t}).
\end{align*}

\subsection{Probabilistic Analysis}\label{sec:ana}
Lemma \ref{lem:det} derives an expression for $\|\wt\|$, for fixed $\g,\h,\qb$. Here, we evaluate the limiting behavior of this expression. Recall that $\g,\h,\qb$ are all i.i.d. standard Gaussian vectors and assume the large-system limit linear regime as in the statement of Theorem \ref{thm:main}. We use the following notation: let $\{X_n\}_{n=1}^\infty$ be sequence of random variables and $\{c_n\}$ a deterministic sequence, then $X_n\rP c_n$ iff for all $\eps>0$, the event $|X_n-c_n|\leq\eps c_n$ occurs w.p. 1 in the limit $n\rightarrow\infty$. For the purpose of this section, convergence is to be understood in the aforementioned meaning.

From standard concentration results on Gaussian r.v.s.: $\|\g\|^2\rP m$, $\|\h\|^2\rP n$, $\|\qb\|^2\rP n$, $\|\beta\h-\sigma\qb\|^2\rP (\beta^2+\sigma^2)n$, and $\ome\rP\sigma_{\tilde\fb}\DelxfR$. For the last relation, we have used the property of the gaussian width as in\cite[Prop.~10.1]{TroppEdge}. Hence, for any fixed $\beta,t$, the objective function in \eqref{eq:goropt} converges to 
\begin{align}\label{eq:uni}
d(\beta,t) =  t +  \frac{\beta^2(m-t)}{\beta^2m+t} n - \frac{t^2+\beta^2m^2}{m(\beta^2m+t)}\Delxf
\end{align}
It can be checked that the objective function in \eqref{eq:goropt} is convex in $t$ and concave in $\beta$. Also, the constraint sets are compact. Thus, it follows from \cite[Cor.~II.1]{AG1982} ( ``point-wise convergence in probability of concave functions implies uniform convergence in compact spaces" ) that the convergence in \eqref{eq:uni} is uniform over $\beta$ and $t$. As will be shown next, provided that the constants determining the constraint sets are large enough, then there exist unique $\beta_*^2$ and $t_*$ that are optimal in \eqref{eq:uni}. Hence, as in  \cite[Thm.~2.7]{NF36},  the optimal solutions of \eqref{eq:goropt} indeed converge to the deterministic solutions of \eqref{eq:uni}, which we calculate below. Let the constant bounds on the variables $\beta,t$, namely $\La,T$, to be specified later. Denote $\beta_*,t_*$ optimal solutions in
$$
\max_{0\leq\beta\leq B}\min_{0\leq t\leq T} d(\beta, t).
$$
Let us write $\omega:=\DelxfR$. We differentiate the objective with respect to both $\beta$ and $t$ to find:
\begin{subequations}
\begin{align}
&\frac{\partial d(\beta_*,t_*)}{\partial\la} = 1 - \frac{\beta_*^2m(\beta_*^2-1)}{(t_*+\beta_*^2m)}n - \frac{t_*^2+2t_*\beta_*^2m-m^2\beta_*^2}{(\beta_*^2m+t_*)^2}\frac{\omega^2}{m}\label{eq:la},\\
&\frac{\partial d(\beta_*,t_*)}{\partial\beta} = \frac{2\beta_*t_*(m-t_*)}{(\beta_*^2m+t_*)^2}(n-\omega^2) \label{eq:beta}.
\end{align}
\end{subequations}
Setting them to zero, from \eqref{eq:beta} we have $\beta_*=0$, $t_*=0$ or $t_*=\sigma m$. We consider each case separately. Assume $\beta_*=0$, then $t_* = \arg\min d(\la,\beta_*) = \arg\min \la - \la\frac{\omega^2}{m} = 0$ and $d(t_*,\beta_*)=0$. Next, suppose $t_* =  m$.  Substituting this in \eqref{eq:la} we find
\begin{align}
(n-m)\beta_*^4 + ((n-m)-(m-\omega^2))\beta_*^2 - (m-\omega^2) = 0.\nn
\end{align} 
Solving this, yields $\beta_*^2 = \frac{m-\omega^2}{n-m}$ and $d(\beta_*,t_*)=m-\omega^2>0$. Choose, $\La,T$ such that $\beta_*,t_*$ are feasible. Form convexity, first-order optimality conditions are sufficient.  

What is left is to substitute those limit values $\beta_*$ and $t_*$ in \eqref{eq:wopt} in Lemma \ref{lem:det}, to conclude with
$$
{\|\wt\|^2}/{\sigma^2}\rP{\Delxf(n-\Delxf)}/{m-\Delxf}.
$$

\subsection{Proof Outline of Theorem \ref{thm:main2}}\label{sec:malakizomai}
In the next few lines we outline only the main checkpoints involved in the proof of Theorem \ref{thm:main}. The analysis follows along the same lines as in Sections \ref{sec:pre}, \ref{sec:gor} and \ref{sec:ana} for the proof of Theorem \ref{thm:main}. In fact, things here are less involved since we are only interested in lower bounding the optimal cost of a min-max problem, and don't care about its optimal values. Hence, a single application of GMT, and not the framework of \cite{StoLASSO,tight} requires to be employed.  A detailed proof will be included in the future extended version of the paper.

Denote, $\Cc:=\Tc_f(\x_0)\cap\Sc^{n-1}$. We write the $\mcsv$ of $\A$ as
\begin{align}\label{eq:start}
\smin = \min_{\x\in\Cc} \max_{\|\y\|\leq 1} \y^T\A\x.
\end{align}
We prove a high-probability lower bound on the optimal cost of this min-max optimization. We do so by applying Gordon's GMT, just as is done in the Gaussian case. But first, we need to bring \eqref{eq:start} in a format where GMT is applicable
After replicating the ideas of Section \ref{sec:pre} and applying GMT, it can be shown that it suffices to lower bound the optimal cost of the following (AO) problem instead:
\begin{align}\label{eq:start2}
\min_{\x\in\Cc,\ellb} \max_{B\geq \beta\geq 0} \|\x-\ellb\| + \beta(\|\ellb\|\|\g\|-\h^T\ellb).
\end{align}
Next, as in Section \ref{sec:gor} we perform a deterministic (fixed $\g,\h$) analysis of this to simplify it as possible into a scalar optimization problem. Only caution should be taken here that the constraint set $\Cc$ on $\x$ is non-convex, thus we are not allowed to flip min-max operations ``carelessly". It can be shown that \eqref{eq:start2} has optimal cost $\sqrt{F}$, where $F:=F(\g,\h)$ is the optimal cost of the following optimization:
\begin{align}\label{eq:start2}
\min_{\x\in\Cc,T\geq t\geq 0} \max_{B\geq \beta\geq 0} \frac{\|\g\|^2-t\beta^2\|\h\|^2-2t\beta \h^T\x + t\beta^2\|\g\|^2 + t^2\beta^2}{\|\g\|^2+t}.
\end{align}
After applying the min-max inequality (e.g. \cite[Lemma~36.1]{Roc70}) it is easy to optimize over $\x$ by choosing it to maximize $\h^T\x$ in $\Cc$ and $F$ is the optimal cost to only a scalar optimization problem involving the r.v.s. $\|\g\|$, $\|\h\|$ and $\max_{\x\in\Cc}\h^T\x$. All three, are 1-Lipschitz functions, thus, they concentrate (thus, converge in the proportional regime) to their mean values $\sqrt{m}$, $\sqrt{n}$ and $\DelxfR$ respectively. Also, the problem is convex in $\beta$, t, thus we can yield the expression of Theorem \ref{thm:main2} (with the correspondence $\beta\leftrightarrow \chi$, $t\leftrightarrow\rho$), by first-order optimality conditions.